\documentclass{aamas2017}
\pdfoutput=1
\usepackage{times}
\usepackage{enumerate,setspace}
\usepackage{amsmath,amssymb,amsfonts,ascmac}
\usepackage{graphicx}
\usepackage{algorithm}
\usepackage{algorithmic}
\usepackage{enumitem}
\setitemize{noitemsep,topsep=1pt,parsep=1pt,partopsep=1pt}
\usepackage{helvet}
\usepackage{courier}
\usepackage{booktabs,microtype}
\usepackage{tikz}
\usetikzlibrary{arrows,shapes, calc, fit, positioning}
\usepackage{tkz-graph}
\usepackage{xcolor,colortbl}
\usepackage{booktabs,microtype}
\usepackage{multirow}
\newcommand{\bbR}{\mathbb{R}}

\newcommand{\calH}{\mathcal{H}}
\newcommand{\calP}{\mathcal{P}}

\newcommand{\calG}{\mathcal{G}}
\newcommand{\poly}{\mathrm{poly}}
\newcommand{\GASP}{\scalebox{0.9}{\sf GASP}}
\newcommand{\gGASP}{\scalebox{0.9}{\sf gGASP}}

\newtheorem{theorem}{Theorem}

\newtheorem{proposition}[theorem]{Proposition}
\newtheorem{cor}[theorem]{Corollary}
\newtheorem{observation}[theorem]{Observation}
\newtheorem{definition}{Definition}
\newtheorem{example}{Example}

\begin{document}
\title{On Parameterized Complexity of Group Activity Selection Problems on Social Networks}
\numberofauthors{3}
\author{
\alignauthor
Ayumi Igarashi\\
       \affaddr{University of Oxford}\\
      \affaddr{Oxford, United Kingdom}\\
       \email{ayumi.igarashi@cs.ox.ac.uk}
\alignauthor
Robert Bredereck\\
       \affaddr{University of Oxford}\\
       \affaddr{Oxford, United Kingdom}\\
       \email{robert.bredereck@tu-berlin.de}
\alignauthor
Edith Elkind\\
      \affaddr{University of Oxford}\\
      \affaddr{Oxford, United Kingdom}\\
       \email{elkind@cs.ox.ac.uk}
}
\maketitle

\begin{abstract}
In Group Activity Selection Problem (\GASP), players form coalitions to participate in activities
and have preferences over pairs of the form (activity, group size).
Recently, Igarashi et al.~\cite{Igarashi2016b} have initiated the study of group activity selection problems on social 
networks (\gGASP): a group of players can engage in the same activity 
if the members of the group form a connected subset of the underlying 
communication structure. Igarashi et al.~have primarily focused on Nash stable outcomes, and showed
that many associated algorithmic questions are computationally hard even for very simple networks.
In this paper we study the parameterized complexity of \gGASP\ with respect to the number of activities
as well as with respect to the number of players,
for several solution concepts such as Nash stability, individual stability and core stability.
The first parameter we consider in the number of activities. For this parameter, 
we propose an FPT algorithm for Nash stability for the case where the social network is acyclic
and obtain a W[1]-hardness result for cliques (i.e., for classic \GASP); similar results
hold for individual stability. In contrast, finding a core stable outcome is hard even
if the number of activities is bounded by a small constant, both for classic \GASP\
and when the social network is a star. Another parameter we study
is the number of players. While all solution concepts we consider
become polynomial-time computable when this parameter is bounded by a constant,
we prove W[1]-hardness results for cliques (i.e., for classic \GASP). 
\end{abstract}


\section{Introduction}
\noindent
In mutliagent systems, agents form coalitions to perform tasks. A useful model for analyzing how tasks can 
be allocated to groups of agents is the {\em group activity selection problem} (\GASP), proposed by 
Darmann et al.~\cite{Darmann2012}. In \GASP s, participants express preferences over 
pairs of the form (activity, group size).
The activities are then assigned to participants so as to achieve the best performance for
the whole system as well as to satisfy individual agents. 
The key idea behind this formulation is that ideal group size depends on the task at hand: 
in a company, an ideal size of the sales team may differ from that of a web developers' team.

However, there is one important feature missing from the standard \GASP\
model, namely the {\em feasibility} of resulting groups. In 
many real-life scenarios, smooth communication among members of a group is crucial in order for different individuals to 
work together, and hence one needs to take into account communication structures among agents. For instance, a group of 
employers are unable to realize their full potential if no agent knows each other. Nevertheless, the basic \GASP\ framework 
imposes no restrictions on how agents can split into different groups.

A succinct way to capture such restrictions is to represent communication structures by undirected graphs. This idea 
dates back to cooperative games with graph structure, proposed by Myerson~\cite{Myerson1977}. Under Myerson's proposal, 
nodes of the graph correspond to players and edges represent communication links between them.
Recently, Igarashi et al.~\cite{Igarashi2016b} considered group activity selection on social networks (\gGASP), where 
groups need to be connected in their underlying social network in order to achieve certain tasks.
The focus of that work was on core stability and Nash stability. In contrast with similar settings in 
cooperative games \cite{Elkind2014,Chalkiadakis2016,Igarashi2016a}, many of the computational problems were shown to be 
NP-complete for very simple network structures; in particular, deciding the existence of core and Nash stable outcomes 
was shown to be NP-complete even when the social network is either a path, a star, or has connected components of size at 
most four.

Motivated by the work of Igarashi et al.~\cite{Igarashi2016b}, we investigate the parameterized complexity of finding 
stable outcomes in group activity selection problems. A problem is fixed parameter tractable (FPT) with respect 
to a parameter $k$ if each instance $I$ of this problem can be solved in time $O(f(k)) \poly(|I|)$ where $f$ is 
a computable function. We show that the problem of deciding the existence of Nash stable outcomes for \gGASP s on 
acyclic graphs is fixed parameter tractable with respect to the number of activities, thereby solving a problem left 
open by Igarashi et al.~\cite{Igarashi2016b}. For general graphs, we obtain a W[1]-hardness result, 
implying that this problem is unlikely to admit an FPT algorithm; 
in fact, our hardness result holds even for \gGASP s on cliques (which correspond to the classic \GASP).

We then consider a weaker stability concept, namely, {\em individual stability}. In the context of hedonic games, it is 
known that an individually stable solution always exists and can be computed in polynomial time when the communication 
structure is acyclic~\cite{Igarashi2016a}. In contrast, we show that an instance 
of \gGASP\ may have no individually stable outcomes, even if the underlying network is a path.
We then analyze the complexity of computing individually stable solutions, 
and discover that from an algorithmic point of view 
individual stability is very similar to Nash stability.

\begin{table*}[t] 
	\footnotesize
	\caption{Overview of our complexity results. `NS' stands for Nash stability, 
                 `IS' stands for individual stability, `CR' stands for core stability.
                 All W[1]-hardness results are accompanied by XP-membership proofs except 
                 for finding an individually stable assignment in {\sf GASP}s on cliques.
                 For all `XP'-entries, the question whether the problem is fixed-parameter tractable remains open.
                 $\diamondsuit$ indicates that the result is not directly stated, but follows indirectly.
                 New results are listed in boldface font.}
	\tabcolsep=0.6mm
	\centering
	\begin{tabular}{llcccc}
		\toprule
		&& Complexity (general case) & few activities (FPT wrt $p$)& few players (FPT wrt $n$)&copyable activities \\
		\midrule
		\midrule
		NS
		& cliques & NP-c. \cite{Darmann2015} & \bf W[1]-h. (Th. \ref{thm:W1:clique:NS})& \bf W[1]-h. (Th. \ref{thm:W1players:clique}) & NP-c. \cite{Ballester2004}\\
		& acyclic & NP-c. \cite{Igarashi2016b}& \bf FPT (Th. \ref{thm:FPT:tree:NS}) & \bf XP (Obs. \ref{obs:XPplayers}) & poly time \cite{Igarashi2016b}\\
		& paths & NP-c.  \cite{Igarashi2016b}& FPT \cite{Igarashi2016b} & \bf XP (Obs. \ref{obs:XPplayers}) & poly time \cite{Igarashi2016b}\\
		& stars & NP-c.  \cite{Igarashi2016b}& FPT \cite{Igarashi2016b} & \bf XP (Obs. \ref{obs:XPplayers}) & poly time \cite{Igarashi2016b}\\
		& small components & NP-c.\cite{Igarashi2016b} &FPT \cite{Igarashi2016b}& \bf XP (Obs. \ref{obs:XPplayers}) & poly time \cite{Igarashi2016b}$\diamondsuit$\\
		\midrule
		IS 
		& cliques & NP-c. \cite{Darmann2015} & \bf W[1]-h. (Th. \ref{thm:W1:clique:IS})& \bf W[1]-h. (Th.  \ref{thm:W1players:clique}) & NP-c. \cite{Ballester2004}\\
		& acyclic & \bf NP-c. (Th. \ref{thm:NP:IS}) & \bf FPT (Th. \ref{thm:FPT:tree:IS})& \bf XP (Obs. \ref{obs:XPplayers}) & \bf poly time (Th. \ref{thm:IS:copyable})\\
		& paths & \bf NP-c. (Th. \ref{thm:NP:IS}) & \bf FPT (Th. \ref{thm:FPT:tree:IS})& \bf XP (Obs. \ref{obs:XPplayers}) & \bf poly time (Th. \ref{thm:IS:copyable})\\
		& stars & \bf NP-c. (Th. \ref{thm:NP:IS}) & \bf FPT (Th. \ref{thm:FPT:tree:IS})& \bf XP (Obs. \ref{obs:XPplayers}) & \bf poly time (Th. \ref{thm:IS:copyable})\\
		& small components & \bf NP-c. (Th. \ref{thm:NP:IS}) & FPT \cite{Igarashi2016b}$\diamondsuit$ & \bf XP (Obs. \ref{obs:XPplayers}) &  poly time \cite{Igarashi2016b}$\diamondsuit$\\
	\midrule
		CR
		& cliques & NP-c. \cite{Darmann2015} & \bf NP-c. for $\boldsymbol{p=4}$ (Th. \ref{thm:NP:clique:core})& \bf W[1]-h. (Th. \ref{thm:W1players:clique}) & NP-c. \cite{Ballester2004}\\	
		& acyclic & NP-c. \cite{Igarashi2016b} & \bf NP-c. for $\boldsymbol{p=2}$ (Th. \ref{thm:NP:star:core})& \bf XP (Obs. \ref{obs:XPplayers})& poly time \cite{Demange2004}\\
		& paths & NP-c. \cite{Igarashi2016b} & \bf XP (Cor. \ref{cor:core-xp-path})  & \bf XP (Obs. \ref{obs:XPplayers})&poly time \cite{Demange2004}\\
		& stars & NP-c. \cite{Igarashi2016b} & \bf NP-c. for $\boldsymbol{p=2}$ (Th. \ref{thm:NP:star:core}) & \bf XP (Obs. \ref{obs:XPplayers})& poly time \cite{Demange2004}\\
		& small components & NP-c. \cite{Igarashi2016b} &FPT \cite{Igarashi2016b}& \bf XP (Obs. \ref{obs:XPplayers})& poly time \cite{Igarashi2016b}$\diamondsuit$\\		
		\bottomrule
	\end{tabular}
	\vspace{-5pt}
	\label{table}
\end{table*}

Unfortunately, our FPT results do not extend to core stability: we prove that checking the existence of core stable 
assignments is NP-complete even for \gGASP s on stars with two activities; for classic \GASP, we can prove
that this problem is hard if there are at least four activities.
On the other hand, if there is only one activity, 
a core stable assignment always exists and can be constructed efficiently, for any network structure.

Another parameter we consider is the number of players. Somewhat surprisingly, we 
show that this parameterization does not give rise to an FPT algorithm for \gGASP s on general networks. 
Specifically, for all stability notions we consider, 
it is W[1]-hard to decide the existence of a stable outcome even when the underlying graph is a clique. 
We summarize our complexity results in Table~\ref{table}.


\section{Model}
\noindent
For $s,t\in{\mathbb N}\cup\{0\}$ where $s \leq t$, let $[s]=\{1,2,\ldots,s\}$ and $[s,t]=\{s,s+1,\ldots,t\}$. 
We start by defining a group activity selection problem with a graph structure~\cite{Igarashi2016b}.

\begin{definition}[\gGASP]
An instance of the {\em Group Activity Selection Problem with graph structure (\gGASP)} is given by 
a finite set of {\em players}, or {\em agents}, $N=[n]$,
a finite set of {\em activities} $A=A^{*}\cup\{a_{\emptyset}\}$ where 
$A^{*}=\{a_{1},a_{2},\ldots,a_{p}\}$ and $a_{\emptyset}$ is the {\em void activity}, 
a {\em profile} $(\succeq_{i})_{i \in N}$ of complete and transitive preference relations 
over the set of {\em alternatives} $X=A^{*} \times [n]\cup \{(a_{\emptyset},1)\}$, 
as well as a set of communication links between players 
$L \subseteq \{\, \{i,j\} \mid i,j\in N \land i\neq j \,\}$. 
\end{definition}
In what follows, we will write $x \succ_i y$ or $i: x\succ y$ to indicate that player
$i$ strictly prefers alternative $x$ to alternative $y$;
similarly, we will write $x\sim_i y$ or $i: x\sim y$ if $i$ is indifferent between $x$ and $y$.
Also, given two sets of alternatives $X, Y$ and a player $i$, we write
$X\succ_i Y$ to indicate that $i$ is indifferent among all alternatives in $X$
as well as among all alternatives in $Y$, and prefers each alternative in $X$
to each alternative in $Y$.

Two non-void activities $a$ and $b$ are said to be {\em equivalent} if for each player $i \in N$ and every $\ell \in [n]$ 
it holds that $(a,\ell) \sim_{i} (b,\ell)$. A non-void activity $a \in A^{*}$ is called {\em copyable} 
if $A^{*}$ contains at least $n$ activities that are equivalent to $a$ (including $a$ itself). 
Player $i \in N$ is said to {\em approve} an alternative $(a,k)$ if $(a,k) \succ_{i} (a_{\emptyset},1)$. 
When describing a player's preferences, by convention we only list the alternatives that she approves
as well as $(a_{\emptyset}, 1)$.

An outcome of a \gGASP\ is a {\em feasible assignment} of activities $A$ to players $N$, 
i.e., a mapping $\pi:N \rightarrow A$ where for each $a \in A^*$
the set $\pi^a=\{\, i \in N \mid \pi(i)=a \,\}$ of players 
assigned to $a$ is connected in $(N, L)$. Note that we place no constraints 
on the set of players who can simultaneously
engage in the void activity (i.e., do nothing).
For $i \in N$ with $\pi(i)\neq a_\emptyset$, we let $\pi_{i}=\{i\}\cup\{\, j \in N \mid \pi(j)=\pi(i)\}$ 
denote the set of players assigned to the same activity as player $i \in N$; 
for $i \in N$ with $\pi(i)= a_\emptyset$, we set $\pi_i=\{i\}$.

A feasible assignment $\pi:N \rightarrow A$ of a \gGASP\ is {\em individually rational} (IR) if each player weakly prefers 
her own activity to doing nothing, i.e.$(\pi(i),|\pi_{i}|) \succeq_{i} (a_{\emptyset},1)$ for all $i \in N$.

A subset $S\subseteq N$ of players is called a {\em coalition}, and said to be {\em feasible} if it is connected in 
the graph $(N,L)$.

A feasible coalition $S \subseteq N$ and an activity $a \in A^{*}$ {\it strongly block} 
an assignment $\pi:N \rightarrow A$ if
$\pi^a\subseteq S$ and $(a,|S|) \succ_{i} (\pi(i),|\pi_{i}|)$ for all $i \in S$.
A feasible assignment $\pi:N \rightarrow A$ of a \gGASP\ is called {\em core stable} (CR) 
if it is individually rational, and 
there is no feasible coalition $S \subseteq N$ and activity $a \in A^{*}$ 
such that $S$ and $a$ strongly block $\pi$.
Given a feasible assignment $\pi:N \rightarrow A$ of a \gGASP, a player $i \in N$ is said to have 
\begin{itemize}
\item an {\em NS-deviation} 
to activity $a \in A^{*}$ if $\pi^a\cup\{i\}$ is connected, 
and $i$ strictly prefers to join the group $\pi^a$, i.e., $(a,|\pi^a|+1) \succ_{i} (\pi(i),|\pi_{i}|)$.
\item an {\em IS-deviation} if it is an NS-deviation, and all players in $\pi^a$ accept it, 
i.e., $(a,|\pi^a|+1) \succeq_{j} (a,|\pi^a|)$ for all $j \in \pi^a$.
\end{itemize}
A feasible assignment $\pi:N \rightarrow A$ of a \gGASP\ is called {\em Nash stable} (NS) 
(respectively, {\em individually stable} (IS)) 
if it is individually rational and no player $i\in N$ has an NS-deviation (respectively, an IS-deviation) 
to some $a \in A^{*}$.

\section{Nash Stability}\label{sec:NS}
\noindent
The first stability concept we will consider is Nash stability. 
The following example, due to Igarashi et al.~\cite{Igarashi2016b}, 
shows that Nash stable outcomes may not exist 
even in very simple instances of \gGASP.

\begin{example}[Stalker game]\label{ex:NS:empty}
{\em
We consider an instance of \gGASP\ with two players: one is a loner and another is a stalker. 
The loner only approves alternatives of the form $(a, 1)$, where $a\in A^*$,
and the stalker only approves alternatives of the form $(a, 2)$, where $a\in A^*$.
Clearly, no assignment for this instance is Nash stable: if a loner engages
in an activity alone, the stalker would wish to join, and if the loner and the stalker
are together, the loner prefers to deviate to~$a_\emptyset$.
}
\end{example}

Igarashi et al.~\cite{Igarashi2016b} show that it is NP-complete 
to decide if a given instance of \gGASP\ admits a Nash stable outcome,
even if the underlying graph $(N, L)$ is a path or a star;
on the other hand, they demonstrate that for paths and stars
the problem of finding a Nash stable outcome is fixed-parameter
tractable with respect to the number of activities. 
We will now show that this FPT result extends to arbitrary acyclic networks.

\begin{theorem}\label{thm:FPT:tree:NS}
The problem of deciding whether an instance of $\gGASP$ with $|A^*|=p$ 
whose underlying social network $(N, L)$ is acyclic
has a Nash stable feasible assignment and finding one if it exists
is in {\sc FPT} with respect to $p$.
\end{theorem}
\begin{proof}
We will first present a proof for the case where $(N, L)$ is a tree; in the end,
we will show how to extend the result to arbitrary forests.
Fix an instance $(N, A, (\succeq_i)_{i\in N}, L)$ of \gGASP\ such that $(N, L)$ is a tree.
We choose an arbitrary node in $N$ as the root of this tree, thereby making $(N, L)$
a rooted tree; we denote by $C(i)$ the set of children of $i$ and by $D(i)$
the set of descendants of $i$ (including $i$ herself).

We process the nodes from the leaves to the root.
For each $i\in N$, each $B\subseteq A^*$, each $B^\prime\subseteq B$,
each $(a,k) \in B^{\prime} \times [n] \cup \{(a_{\emptyset},1)\}$ and each $t\in[k]$,
we let $f_{i}(B,B^{\prime},(a,k),t)$ be {\em true} if there is an assignment $\pi:N\to A$
where 
\begin{itemize}
\item the set of activities assigned to players in $D(i)$ is exactly $B^{\prime}$;
\item player $i$ is assigned to $a$ and is in a coalition with $k$ other players;
\item exactly $t$ players in $D(i)$ belong to the same group as $i$;
\item the $t$ players in $D(i)\cap \pi_i$ 
      weakly prefer $(a,k)$ to $(b, 1)$ for each $b\in A\setminus B$, 
      and have no incentive to deviate to the other groups, 
      i.e., every player in $D(i)\cap \pi_i$ whose children do not belong to $\pi_i$ likes $(a, k)$ 
      at least as much as each of the alternatives she can deviate to; 
\item the players in $D(i)\setminus \pi_i$ weakly prefer their alternative under $\pi$ 
      to engaging alone in any of the activities in $A\setminus B$, 
      have no NS deviation to activities in $B^{\prime}\setminus \{a\}$, 
      and have no incentive to deviate to $i$'s coalition, i.e., if $a \neq a_{\emptyset}$, 
      then every player $j\in D(i)\setminus \pi_i$ whose parent belongs to $\pi_i$ 
      likes $(\pi(j), |\pi_j|)$ at least as much as $(a,k+1)$.
\end{itemize}
Otherwise, we let $f_{i}(B,B^{\prime},(a,k),t)$ be \emph{false}.

By construction, our instance admits a Nash stable assignment
if and only if $f_{r}(B,B^{\prime},(a,k),k)$ is {\em true}
for some combination of the arguments $B, B^\prime$, and $(a, k)$,
where $r$ is the root.

If $i$ is a leaf, we set $f_{i}(B,B^{\prime},(a,k),t)$ to \emph{true} if 
$B^{\prime}=\{a\}$, $t=1$, and $i$ weakly prefers $(a,k)$ to every alternative $(b, 1)$ 
such that $b\in A\setminus B$; otherwise, we set $f_{i}(B,B^{\prime},(a,k),t)$ to \emph{false}.
Now, consider the case where $i$ is an internal vertex. 
We first check whether $i$ strictly prefers some alternative $(b,1)$ such that $b\in A\setminus B$ to $(a,k)$; 
if so, we set $f_{i}(B,B^{\prime},(a,k),t)$ to \emph{false}. Otherwise, we proceed and check for each partition $\calP$ 
of $B^{\prime} \setminus \{a\}$ whether there is an allocation of each activity set $P\in \calP$ 
to some subtree rooted at $i$'s child that gives rise to an assignment with the conditions described above. 

We do this by using the color-coding technique; we `color' each child of $i$ using colors from $\calP$ independently 
and uniformly at random.
Suppose that $\pi$ is an assignment satisfying the properties described above where each activity set $P\in \calP$ 
is assigned to $D(j)$ for some $j\in C(i)$; we denote by $S$ the set of $i$'s children whose subtrees 
are assigned to some $P\in \calP$. Then, the probability that the subtrees associated with $S$ 
are assigned to these activities by a coloring $\chi$ chosen at random is $|\calP|^{-|\calP|}$, 
since there are $|\calP|^{|C(i)|}$ possible colorings, and $|\calP|^{|C(i)|-|\calP|}$ of them coincide with $\pi$ on $S$. 
We can then derandomize our algorithm using a family of $k$-perfect hash functions~\cite{Alon1995}.

Now, let us fix a coloring $\chi:C(i) \rightarrow \calP$. For each $P \in \calP$, 
we denote by $C_{P}=\{\, j \in C(i) \mid \chi(j)=P\,\}$ the set of $i$'s children of color $P$. 
We seek to assign each subtree rooted at the $j$-th element of $C(i)$ 
to the activities $\chi(j)\cup \{a, a_{\emptyset}\}$ in such 
a way that exactly one subtree of each color $P \in \calP$ uses all the activities in $P$.
We will show that there exists an efficient algorithm that finds an assignment compatible with $\chi$, 
or determines that no such assignment exists.

We will first determine for each color $P \in \calP$, each $j \in C_P$, and each $\ell\in[0, k-1]$ whether the subtree 
rooted at $j$ can be assigned to the activity set $P$, and exactly $\ell$ players in the subtrees of $C_P$ can be assigned 
to $a$; we refer to this subproblem by $f_P(j,\ell)$.
We initialize $f_P(j,\ell)$ to {\em true} if 
\begin{itemize}
\item $\ell=0$, and $i$ and $j$ can be {\em separated} from each other, that is, there exists an alternative $(b,\ell)\in P 
\times [n]\cup \{(a_{\emptyset},1)\}$ such that (i.) $f_{j}(B,P,(b,\ell),\ell)$ is {\em true}, (ii.) $b=a_{\emptyset}$ or $i$ weakly 
prefers $(a,k)$ to $(b,\ell+1)$, and (iii.) $a=a_{\emptyset}$ or $j$ weakly prefers $(b,\ell)$ to $(a,k+1)$; or
\item $\ell\geq 1$, and $f_{j}(B,P\cup \{a\},(a,k),\ell)$ is {\em true}.
\end{itemize}
We set $f_P(j,\ell)$ to {\em false} otherwise. 
Then, we iterate through all the subtrees associated with players in $C_P\setminus \{j\}$ 
and update $f_{P}(j,\ell)$: for each $j^{\prime} \in C_P \setminus \{j\}$ 
and for $\ell=k-1,k-2,\ldots,0$, we set $f_{P}(j,\ell)$ to {\em true} if
\begin{itemize}
\item $\ell \geq 2$, and, moreover, there exists an $x \in [\ell]$ such that both $f_{j^{\prime}}(B,\{a\},(a,k),x)$ and 
$f_{P}(j,\ell-x)$ are {\em true}; or
\item both $f_{j^{\prime}}(B,\emptyset,(a_{\emptyset},1),1)$ and 
$f_{P}(j,\ell)$ are {\em true}, 
and $a=a_{\emptyset}$ or player $j^{\prime}$ weakly prefers $(a_{\emptyset},1)$ to $(a,k+1)$;
\end{itemize}
we set $f_{P}(j,\ell)$ to {\em false} otherwise.

In a similar manner, we determine whether exactly $\ell$ successors of $i$ can be assigned to activity $a$; 
we denote this problem by $f(\ell)$. We initialize $f(1)$ to {\em true} and $f(\ell)$ to {\em false} 
for each $\ell \in [2, t]$. We iterate through all the colors in $\calP$ and update $f(\ell)$ one by one; 
for each color $P \in \calP$ and each number $\ell=t,t-1,\ldots,1$, we set $f(\ell)$ to {\em true} 
if there exists $x \in [0, t-1]$ such that $f_P(j,x)$ is true for some $j \in C_P$ 
and $f(\ell-x)$ is {\em true}, and we set $f(\ell)$ to {\em false} otherwise. 
Finally, we reject the coloring $\chi$ if $f(t)$ is {\em false}. 
Clearly, if there exists a Nash stable feasible assignment that is compatible with $\chi$, 
then the algorithm does not reject the coloring. We omit the proof for the bound on the running time.

Now, if $(N, L)$ is a forest, we can combine the algorithm described above with the algorithm
for graphs with small connected components proposed by Igarashi et al.~\cite{Igarashi2016b}.
The running time of the latter algorithm is a product of the time required to find a Nash stable assignment
for a single connected component and the time required to combine solutions for different components;
in the analysis of Igarashi et al., the former is $O(p^c)$, where $c$ is the maximum component size and the latter is $O(8^pn^3)$.
In our case, each connected component is a tree, so instead of the $O(p^c)$ algorithm for general graphs
we can use our FPT algorithm for trees. This shows that our problem is in FPT for arbitrary forests.
\end{proof}

It is unlikely that Theorem~\ref{thm:FPT:tree:NS} can be extended to general graphs
or even cliques: 
our next result shows that the problem of finding a Nash stable 
outcome is W[1]-hard with respect to the number of activities
even for `vanilla' \GASP, i.e., when the social network imposes no constraints
on possible coalitions.

\begin{theorem}\label{thm:W1:clique:NS}
The problem of determining whether an instance of \gGASP\ 
admits a Nash stable assignment 
is {\em W[1]}-hard with respect to the number of activities,
even if the underlying graph $G=(N, L)$ is a clique.
\end{theorem}
\begin{proof}
We reduce from {\sc Clique} on regular graphs, which is known to be W[1]-hard 
(see e.g. Theorem 13.4 in the book of Cygan~\cite{Cygan2015}). 
Given a graph $G=(V,E)$ and an integer $k$, where $|V|=n$, $|E|=m$, and each vertex of $G$ has degree $\delta \geq k-1$, 
we create an instance of \gGASP\ whose underlying graph is a clique, as follows. 

We define the set of activities as
\[
A =A^{\prime}\cup B^{\prime}\cup\{x, a_\emptyset\},
\] 
where $|A^\prime|=k$, $|B^\prime|=k(k-1)/2$.
Notice that $|A^*|=|A|-1 = 1+k(k+1)/2$.

For each $v \in V$, we create a vertex player $v$, and for each edge $e=\{u,v\} \in E$, we create two edge players 
$e_{uv}$ and $e_{vu}$. 
Let $X=\{j(k+3)+2+\delta\mid j\in[n]\}$, and let $\alpha:V\to X$ be a bijection that assigns a 
distinct number in $X$ to each vertex $v\in V$. Note that $\alpha(u)<\alpha(v)$
implies that the intervals $[\alpha(u),\alpha(u)+k+1]$ and 
$[\alpha(v)-1,\alpha(v)+k]$ are disjoint.
Similarly, let $Y=\{1+2j\mid j\in[m]\}$ and 
let $\beta:E \rightarrow Y$ be a bijection that assigns a distinct number
in $Y$ to each edge $e\in E$.
For each $v\in V$ we construct a set of $\alpha(v)-\delta+k-2$ dummy players
$D_v$; 
similarly, for each $e\in E$ we construct a set of $\beta(e)-2$ dummy players
$D_e$.
We let
\[
N_{G}=
V \cup \bigcup_{v\in V} D_v \cup \bigcup_{e=\{u,v\} \in E}\{e_{uv},e_{vu}\} \cup D_e.
\]
Finally, we create five additional players $\{b_1, b_2, c_1, c_2, g\}$ and let
$$
N=N_{G}\cup \{b_1, b_2, c_1, c_2, g\};
$$ 
intuitively, $b_1, b_2$ and $c_1, c_2$ will form two instances
of the stalker game, with $g$ used to stabilize the latter.
Let $L$ be the set of pairs of distinct players of $N$. 
Note that $|A^*|$ depends on $k$, but not on $n$, and the size of our instance of \gGASP\ is bounded by $O(n^3+m^2)$.

We will now define the players' preferences.
Each vertex player $v \in V$ and the players in $D_v$ approve each 
alternative in $A^{\prime}\times [\alpha(v),\alpha(v)+k+1]$.
Each edge player $e_{vu}$ approves each alternative 
in $A^{\prime}\times [\alpha(v),\alpha(v)+k+1]$ as well as each alternative
in $B^{\prime}\times\{\beta(e)\}$, whereas its 
dummies only approve the alternatives in $B^{\prime}\times\{\beta(e)\}$.
All of these players are indifferent among all alternatives they approve.
The stabilizer $g$ approves each alternative of the form $(a, s)$ with $a \in A^{\prime}$, 
$s\in \bigcup_{v \in V}[\alpha(v)+2,\alpha(v)+k+2]$ and is indifferent among them; 
she also approves $(x,3)$, but likes it less than all other approved alternatives.
Also, we have
\begin{align*}
c_1: (x, 1)\sim (x, 3)\succ (a_\emptyset, 1), \ \
&c_2: (x, 2)\sim (x, 3)\succ (a_\emptyset, 1),\\
b_1: A^\prime\times\{1\}\succ (a_\emptyset, 1), \qquad
&b_2: A^\prime\times\{2\}\succ (a_\emptyset, 1).
\end{align*}

We will now argue that the graph $G$ contains a clique of size $k$ if and only if there exists a Nash stable assignment
for our instance of \gGASP. 
Suppose that $G$ contains a clique $S$ of size $k$. We construct an assignment $\pi$ as follows. 
We establish a bijection $\eta$ between $S$ and $A^\prime$, and for each $v\in V$
we form a coalition of size $\alpha(v)$ that engages in $\eta(v)$: this coalition consists
of $v$, all players in $D_v$, and all edge players $e_{vu}$ such that $u\not\in S$. 
Also, we establish a bijection $\xi$ between the edge set $\{\,\{u,v\}\in E\mid u,v \in S\,\}$ 
and $B^{\prime}$, and assign the activity $\xi(e)$ 
to the edge players $e_{uv}$, $e_{vu}$ as well as to all players in $D_e$. 
Finally, we set $\pi(c_1)=\pi(c_2)=\pi(g)=x$,
and assign the void activity to the remaining players.
We will now argue that the resulting assignment $\pi$ is Nash stable. 

Clearly, no player assigned to an activity in $A^\prime$ or $B^\prime$
wishes to deviate.
Now, consider a player $v\in N$ with $\pi(v)=a_\emptyset$;
by construction, $v$ only wants to join a coalition if it engages in an activity in $A^\prime$
and its size is in the interval $[\alpha(v)-1, \alpha(v)+k]$, and no such coalition
exists. The same argument applies to players in $D_v$.
Similarly, consider an edge player $e_{vu}$ with $\pi(e_{vu})=a_\emptyset$.
We have $u, v\not\in S$, and therefore $e_{vu}$ does not want to join
any of the existing coalitions; the same argument applies to all dummies of $e_{vu}$.
Further, $b_1$ and $b_2$ do not want to deviate
since each activity $a \in A^{\prime}$ is assigned to at least three players,
and players $c_1$ and $c_2$ do not want to deviate since they are allocated one of their top choices. 
Finally, the stabilizer $g$ does not want to deviate, since there is no coalition of size 
$s\in \bigcup_{v\in V}[\alpha(v)+1,\alpha(v)+k+1]$ that engages in an activity $a\in A^\prime$. 
Hence, $\pi$ is Nash stable.

Conversely, suppose that there exists a Nash stable feasible assignment $\pi$.
Notice that $\pi$ cannot allocate an activity $a \in A^{\prime}$ to $b_1$ or $b_2$, 
or leave it unallocated, since no such assignment can be Nash stable.
Thus, each activity in $A^\prime$ is allocated to a coalition 
whose size lies in the interval $[\alpha(v),\alpha(v)+k+1]$ 
for some $v \in V$.
Further, Nash stability implies that $\pi$ allocates $x$ to $c_1$, $c_2$ and $g$.
Now, if some activity $a \in A^{\prime}$ is assigned to $s$ players,
where $s\in[\alpha(v)+1, \alpha(v)+k+1]$ for some $v \in V$, 
the stabilizer $g$ would then deviate to that coalition;
hence, for each $a\in A^\prime$ we have $|\pi^a|=\alpha(v)$ for some $v\in V$. 
Now, let $S=\{v\in V\mid \alpha(v) = |\pi^a|\text{ for some }a\in A^\prime\}$.
By construction, $|S|=k$. We will show that $S$ is a clique. 

Consider a player $v\in S$, and let $a(v)$ be the activity assigned to $\alpha(v)$
players under $\pi$. We have $\pi(v)= a(v)$, since otherwise $v$ could deviate to $a$;
similarly, all players in $D_v$ are assigned to $a$.
The only other players who approve $(a(v), \alpha(v))$ are edge players $e_{vu}$.
Thus, $\delta-k+1$ such players must be assigned to $a(v)$, and the remaining 
$k-1$ of these players must be assigned to some activity $b\in B^\prime$,
since otherwise they would deviate to $a(v)$. 
Now, consider an edge player 
$e_{vu}$ with $\pi(e_{vu})=b$ for some $b\in B^\prime$. By individual rationality
we have $\pi^b=\beta(\{u, v\})$ and hence $\pi^b$ consists of
$e_{vu}$, $e_{uv}$ and all dummies of the edge $\{u, v\}$.

For each $v\in S$, consider the set of players 
$E_{\pi, v} = \{e_{vu}, e_{uv}\mid v\in S, \pi(e_{vu})\neq\pi(v)\}$.
The size of this set is $2(k-1)$, and we have argued that
all players in such sets need to be assigned to activities in $B^\prime$.
However, $|B^\prime|=k(k-1)/2$, and each activity in $B^\prime$
can be assigned to at most two edge players simultaneously. Thus, $|\bigcup_{v\in S}E_{\pi, v}|=k(k-1)$.
As each edge player $e_{uw}\in \bigcup_{v\in S}E_{\pi, v}$ can belong to at most two sets $E_{\pi, v}$, it follows that 
it belongs to exactly two such sets, namely, $E_{\pi, u}$ and $E_{\pi, w}$. Thus, both
$u$ and $w$ are in $S$, i.e., $S$ is a clique, which is what we wanted to prove.
\end{proof}

On the positive side, for \gGASP s on cliques, we can place the problem of finding a Nash stable assignment
in the complexity class XP with respect to the number of activities. However, it is not clear if this result
can be extended to general \gGASP s.

\begin{theorem}\label{thm:XP:NS}
There exist an algorithm that, given an instance $(N, A, (\succeq_i)_{i\in N}, L)$
of \gGASP\ with $|N|=n$, $|A|=p+1$ such that $(N, L)$ is a clique, 
determines whether it admits a Nash stable 
assignment in time $(n+1)^p\poly(n)$.
\end{theorem}
\begin{proof}
For every mapping $f:A^* \rightarrow [0, n]$, we will check if there is a Nash stable assignment such that 
$|\pi^a|=f(a)$ for each $a \in A^*$. There are at most $(n+1)^p$ such mappings; hence, it remains to show that each 
check will take at most $\poly(n)$ steps.

Fix a mapping $f:A^* \rightarrow [0, n]$. We construct an instance of the network flow problem as follows. We 
introduce a source $s$, a sink $t$, a node $i$ for each player $i \in N$, and a node $a$ for each activity $a \in A^*$. 
We create an arc with unit capacity from the source $s$ to each player, and an arc with capacity $f(a)$ from node $a \in A^*$ 
to the sink $t$. Then, for each $i\in N$
we create an arc of unit capacity from player $i$ to an activity $a \in A^*$ if and only if $i$ 
weakly prefers $(a,f(a))$ to $(a_\emptyset, 1)$ and to all pairs of the form $(b, f(b)+1)$, where $b\in A^* \setminus 
\{a\}$. It can be easily verified that an integral flow of size $\sum_{a \in A^*}f(a)$ in this network corresponds to a 
Nash stable assignment where exactly $f(a)$ players are engaged in each activity $a \in A^*$. It remains to note 
that one can check in polynomial time whether a given network admits a flow of a given size.
\end{proof}

\section{Individual stability}
\noindent
We will next consider a less stringent stability requirement, namely, individual stability. 
Igarashi and Elkind~\cite{Igarashi2016a} showed that in the context of hedonic games, 
acyclicity is sufficient for individually stable outcomes to exist: 
an individually stable partition of players always exists and can be computed in polynomial time. 
In contrast, it turns out that for \gGASP s this is not the case:
an individually stable outcome may fail to exist even if the underlying social network is a path;
moreover, this may happen even if there are only three players and their preferences are strict.

\begin{example}\label{ex:empty:IS}
{\em
Consider a \gGASP\ with $N=\{1,2,3\}$, $A^{*}=\{x,y,z\}$, $L=\{\{1,2\},\{2,3\}\}$, 
where players' preferences are as follows:
\begin{align*}
1:&~ (b,2) \succ (a,1) \succ (c,3) \succ (c,2) \succ (c,1) \succ (a_{\emptyset},1)\\
2:&~ (c,3) \succ (c,2) \succ (a,2) \succ (b,2) \succ (b,1) \succ (a_{\emptyset},1)\\
3:&~ (c,3) \succ (a,2) \succ (a,1) \succ (a_{\emptyset},1)
\end{align*}

We will argue that each individually rational feasible assignment $\pi$ admits an IS-deviation.
Indeed, if $\pi(1)=a_\emptyset$ then no player is engaged in $c$ and hence player $1$
can deviate to $c$. Similarly, if $\pi(2)=a_\emptyset$ then no player is engaged in $b$
and hence player $2$ can deviate to $b$. There are $9$ individually rational feasible assignments
where $\pi(1)\neq a_\emptyset$, $\pi(2)\neq a_\emptyset$; for each of them we can find an IS deviation
as follows (we write $i\to x$ to indicate that player $i$ has an IS-deviation to activity $x$): 
\begin{enumerate}[label=(\arabic*)]
\item
$\pi(1)=a$, $\pi(2)=b$, $\pi(3)=a_{\emptyset}$: $1\to b$; 
\item
$\pi(1)=b$, $\pi(2)=b$, $\pi(3)=a_{\emptyset}$: $3\to a$; 
\item
$\pi(1)=b$, $\pi(2)=b$, $\pi(3)=a$: $2\to a$; 
\item
$\pi(1)=c$, $\pi(2)=a$, $\pi(3)=a$: $2\to c$; 
\item
$\pi(1)=c$, $\pi(2)=b$, $\pi(3)=a_\emptyset$: $3\to a$;
\item
$\pi(1)=c$, $\pi(2)=b$, $\pi(3)=a$: $2\to a$;
\item
$\pi(1)=c$, $\pi(2)=c$, $\pi(3)=a_{\emptyset}$: $3\to a$; 
\item
$\pi(1)=c$, $\pi(2)=c$, $\pi(3)=a$: $3\to c$; 
\item
$\pi(1)=c$, $\pi(2)=c$, $\pi(3)=c$: $1\to a$. 
\end{enumerate}
}
\end{example}

In Example~\ref{ex:empty:IS} all activities are distinct. 
On the other hand, if all activities are copyable, 
an individually stable outcome is guaranteed to exist. 
Moreover, 
we can adapt the result of Igarashi and Elkind~\cite{Igarashi2016a} for hedonic games
and obtain an efficient algorithm for computing an individually stable outcome.

\begin{theorem}\label{thm:IS:copyable}
Each instance of \gGASP\ where each activity $a \in A^{*}$ is copyable 
and $(N,L)$ is acyclic admits an individually stable feasible assignment; 
moreover, such as assignment can be found in polynomial time.
\end{theorem}
\begin{proof}[Sketch]
The algorithm is similar to the one for hedonic games~\cite{Igarashi2016a}. 
The basic idea is to create a rooted tree for each connected component 
and construct an assignment for every subtree in a bottom-up manner. 
When considering a subtree rooted at player $i$, we start with the assignment obtained by combining
the previously constructed assignments for the subtrees rooted at the children of $i$
and assigning $i$ to the void activity.
We then let player $i$ join the most preferred activity
among those to which she has an IS deviation. After that we keep adding  
players to $i$'s coalition as long as the resulting coalition remains feasible,
the player being added is willing to move, 
and such a deviation is acceptable for all members of $i$'s coalition. 
Similarly to the proof in Theorem~1 of Igarashi and Elkind~\cite{Igarashi2016a}, 
one can show that the resulting assignment is individually stable.
\end{proof}

In contrast, when activities are not copyable, finding an individually stable
feasible assignment is hard, even for very simple social networks.
Indeed, for every class of simple social networks for which Igarashi et al.~\cite{Igarashi2016b}
show that finding a Nash stable outcome is NP-complete, finding
an individually stable outcome is NP-complete as well. The proof of the following theorem 
is similar to the proofs of the respective results of 
Igarashi et al.~\cite{Igarashi2016b}: essentially, we have to replace each instance
of the stalker game in these proofs with an instance of the game from Example~\ref{ex:empty:IS}.

\begin{theorem}\label{thm:NP:IS}
Given an instance of \gGASP\ whose underlying graph is a path, a star, or 
has connected components whose size is bounded by a constant, 
it is {\em NP}-complete to determine whether it has an individually stable feasible assignment.
\end{theorem}

Moreover, the problem of finding individually stable feasible assignments remains
hard even if the number of activities is small, and even if the social network is a clique.
The proof if similar to that of Theorem~\ref{thm:W1:clique:NS}
and is omitted due to space constraints.

\begin{theorem}\label{thm:W1:clique:IS}
The problem of determining whether an instance of \gGASP\ 
admits an individually stable assignment
is {\em W[1]}-hard with respect to the number of activities,
even if the underlying graph $G=(N, L)$ is a clique.
\end{theorem}

However, just as in the case of Nash stability, if we both restrict 
the structure of the social network and assume that the number of activities is small, 
we can obtain positive algorithmic results. The proof of the next theorem is similar 
to that of Theorem~\ref{thm:FPT:tree:NS}; again, we omit it due to space constraints.

\begin{theorem}\label{thm:FPT:tree:IS}
The problem of deciding whether an instance of $\gGASP$ with $|A^*|=p$
whose underlying social network $(N, L)$ is acyclic
has an individually stable feasible assignment and finding one if it exists
is in {\sc FPT} with respect to $p$.
\end{theorem}

The results for individual stability presented so far indicate that from the complexity perspective
it is very similar to Nash stability. However, it is not clear if the XP algorithm
presented in Theorem~\ref{thm:XP:NS} extends to individual stability. The difficulty is that, 
to determine whether an agent $i$ has an IS deviation to an activity $a$, it is not sufficient
to know how many players engage in $a$: knowing their preferences is important 
to decide whether $i$'s deviation will be vetoed by one of the players currently assigned to $a$.
Another important difference concerns copyable activities and games on acyclic graphs:
in this setting, individually stable outcomes always exist (Theorem~\ref{thm:IS:copyable}),
whereas for Nash stable outcomes this is not the case, as illustrated by the stalker game.


\section{Core stability}
\noindent
Igarashi et al.~\cite{Igarashi2016b} have demonstrated that
the core can be empty even for \gGASP s on paths with $3$ players and $2$ activities.
We reproduce their example below.

\begin{example}[\cite{Igarashi2016b}]\label{ex:core:empty}
{\em
Consider a \gGASP\ with $N=\{1,2,3\}$, $A^{*}=\{a,b\}$, $L=\{\{1,2\},\{2,3\}\}$, 
where agents' preferences $(\succeq_{i})_{i \in N}$ are as follows:
\begin{align*}
1:&~ (b,2) \succ_{1} (a,3) \succ_{1} (a_{\emptyset},1)\\
2:&~ (a,2) \succ_{2} (b,2) \succ_{2} (a,3) \succ_{2} (a_{\emptyset},1)\\
3:&~ (a,3) \succ_{3} (b,1) \succ_{3} (a,2) \succ_{3} (a_{\emptyset},1)
\end{align*}
It can be shown that this instance admits no core stable assignment \cite{Igarashi2016b}.
}
\end{example}

On the positive side, we can show that checking whether a given feasible assignment
is core stable is easy, irrespective of the structure of the social network.
The proposition below generalizes Theorem~11 of Darmann~\cite{Darmann2015}
and Theorem~12 of Igarashi et al.~\cite{Igarashi2016b}.

\begin{proposition}\label{prop:in-core}
Given an instance \,\,$(N, (\succeq_i))_{i\in N}, A, L)$\,\,
of \gGASP\ and a feasible assignment $\pi$ for that instance, 
we can decide in polynomial time whether $\pi$ is core stable. 
\end{proposition}
\begin{proof}
Let $A=A^*\cup\{a_\emptyset\}$ and let $n=|N|$.
By scanning the assignment $\pi$ and the players' preferences, 
we can check whether $\pi$ is individually rational.
Now, suppose that this is the case. Then, for each $a\in A^*$
and each $s\in[n]$ we can check if there is a deviation
by a connected coalition of size $s$ that engages in $a$.
To this end, we consider the set $S_{a, s}$ of all players who strictly
prefer $(a, s)$ to their assignment under $\pi$ and verify
whether $S_{a, s}$ has a connected component of size at least $s$
that contains $\pi^a$; if this is the case, $\pi^a$ (which is itself connected or empty)
could be extended to a connected coalition of size exactly 
$s$ that strongly blocks $\pi$. If no such deviation exists, 
$\pi$ is core stable.  
\end{proof}

However, core stability turns out to be more 
computationally challenging that Nash stability and individual stability
when the number of activities is small: we will now show that 
core stable assignments are hard to find even 
if there are only two activities 
and the underlying graph is a star (and thus one cannot expect an FPT result
with respect to the number of activities for this setting).
Later, we will see that this hardness result can be extended to the case where $(N, L)$ is a clique, 
i.e., to classic \GASP, thereby solving a problem left open by the work
of Darmann~\cite{Darmann2015}.

\begin{theorem}\label{thm:NP:star:core}
It is {\em NP}-complete to determine whether an instance of \gGASP\ 
admits a core stable assignment even when the underlying 
graph is a star and the number of non-void activities is~$2$.
\end{theorem}
\begin{proof}
Our problem is in NP by Proposition~\ref{prop:in-core}.
To establish NP-hardness, we reduce from the NP-complete problem {\sc Hitting Set}~\cite{gj}. 
An instance of {\sc Hitting Set} is a family $\calH=\{V_1,V_2,\ldots,V_m\}$ of subsets of a finite set $V$ 
and an integer $k$ with $k<|V|$. 
It is a ``yes''-instance if $\calH$ admits a {\em hitting set} $V^{\prime}\subseteq V$ 
of size at most $k$, i.e., 
$|V^{\prime}| \leq k$ 
and $V^{\prime}\cap V_i\neq\emptyset$ for each $V_i\in\calH$;
otherwise, it is a ``no''-instance.

Given an instance of {\sc Hitting Set}, we can create three disjoint copies of $(V, \calH)$:
for each $v\in V$ we create elements $x_v$, $y_v$ and $z_v$, 
set $W=\{x_v, y_v, z_v\mid v\in V\}$ and for each $i\in[m]$ we let
$X_i=\{x_v\mid v\in V_i\}$, $Y_i=\{y_v\mid v\in V_i\}$, $Z_i=\{z_v\mid v\in V_i\}$
and set $\calG =\{X_i, Y_i, Z_i\mid i\in[m]\}$.
In what follows, we will use the fact that $(W, \calG)$ admits a hitting set of size $3k$ if and only if 
$(V, \calH)$ admits a hitting set of size~$3$. 

Consider an instance $(V, \calH, k)$ of {\sc Hitting Set} and construct the pair
$(W, \calG)$ as described above. For readability, we renumber the elements of $\calG$
as $W_1, \dots, W_{3m}$. 

We construct an instance of \gGASP\ as follows.
We define the set of activities to be $A^{*}=\{a,b\}$. 
We introduce a center player $c$, two players $s_1$ and $s_2$, and a player $w$ 
for each $w \in W$.
Also, for each 
$i\in[3m]$, we let $t_i=i+|W|+1$ 
and create a set $D_i=\{d^{(1)}_i,d^{(2)}_i,\ldots,d^{(t_i-|W_i|-1)}_i\}$ of dummy players.
We then attach $s_1$, $s_2$, each player $w\in W$
and each of the dummies to the center. 
Formally, the graph $(N,L)$ is given by
\[
N=\{c,s_1,s_2\}\cup W\cup \bigcup^{3m}_{i=1}D_i~\mbox{and}~L=\{\,\{c,x\} \mid x \in N\setminus \{c\} \,\}.
\]
Intuitively, for each $i\in[3m]$
the number $t_i$ is the target coalition size when all players 
in $W_i$ are engaged in activity $b$, together with $c$ and the players in $D_i$.

The agents' preferences over alternatives are defined as follows. 
We let $A^{\prime}=\{a\}\times[4, 3k+3]$. For each $w \in W$, we let 
$B_w=\{b\}\times |\{\, t_i\mid w \in W_i \,\}|$; also, set $B = \bigcup_{w\in W}B_w$.
The preferences of each player $w \in W$ are given by
\begin{align*}
&w:~A^{\prime}\succ B_w \succ (a_{\emptyset},1).
\end{align*}
For each $i\in[3m]$ the dummy players in  $D_i$ only approve the alternative $(b,t_i)$.
The preferences of the center player $c$ are given by
\begin{align*}
&c:~(a,2) \succ (b,2) \succ (a,3) \succ B \succ A^{\prime}\succ (a_{\emptyset},1).
\end{align*} 
Finally, the 
preferences of players $s_1$ and $s_2$ 
are given by
\begin{align*}
&s_{1}:~A^{\prime}\succ (b,2)  \succ (a,3) \succ (a_{\emptyset},1)\\
&s_{2}:~A^{\prime} \succ (a,3) \succ B \succ (b,1)\succ (a,2) \succ (a_{\emptyset},1).
\end{align*}
Note that the preferences of $c$, $s_1$ and $s_2$, when restricted to $A\times[1, 2, 3]$,
form an instance of \gGASP\ with an empty core (Example~\ref{ex:core:empty}).

We will show that $(V, \calH)$ admits a hitting set of size at most $k$ 
if and only if there exists a core stable feasible assignment.

Let $V^{\prime}$ be a hitting set of size at most $k$ in $(V, \calH)$.
Set $W^{\prime}=\{x_v, y_v, z_v\mid v\in V^{\prime}\}$; we have $|W^\prime|\le 3k$,
and, by construction, $W^{\prime}$ is a hitting set for $(W, \calG)$.
Then, we construct a feasible assignment 
$\pi$ by assigning activity $a$ to $c$, $s_1$, $s_2$, and the players 
$w \in W^{\prime}$, and assigning the void activity to the remaining players.
Clearly, $\pi$ is individually rational, since $(a,|W^{\prime}|+3) \in A^{\prime}$. Further, notice that no connected 
subset $S$ together with activity $a$ strongly blocks $\pi$:
every such subset has to contain players $s_1$ and $s_2$, who  
are currently enjoying one of their top alternatives. 
It remains to show that no connected subset $S$ together with activity $b$ strongly blocks $\pi$. Suppose 
towards a contradiction that such a subset $S$ exists; 
as $s_1$ and $s_2$ are not interested in deviating, it must be the case 
that $|S|=t_i$ for some $i\in [3m]$ and hence $S$ consists of agents who approve $(b, t_i)$, i.e., 
$S=\{c\}\cup W_i \cup D_i$ for some $i\in[3m]$. However, since $W^{\prime}$ is a hitting set, 
there is an agent $j\in W^{\prime}\cap W_i$ with $\pi(j)=a$, and this agent prefers $(a, |W^{\prime}|+3)$
to $(b, t_i)$, a contradiction.
Hence, $\pi$ is core stable.

Conversely, suppose that there exists a core stable feasible assignment $\pi$ and let 
$W^{\prime}=\{\, w \in W \mid \pi(w)=a \,\}$. 

We will first argue that $\pi(c)=a$.
Indeed, if $\pi(c)=a_\emptyset$, then $\pi^a=\emptyset$ and agents $c$, $s_1$ and $s_2$ can deviate to $a$.
If $\pi(c)=b$ and the only other agent to engage in $b$ is $s_1$, then $\pi^a=\emptyset$
and $c$ and $s_2$ can deviate to $a$. If $\pi(c)=b$ and $|\pi^b|=t_i$ for some $i\in [m]$
then $\pi^a=\emptyset$, and $c$, $s_1$ and $s_2$ can deviate to $a$.
It follows that $\pi(c)=a$.

Now, if $\pi^a=\{c, s_2\}$ then $\pi^b=\emptyset$ and agent $s_2$ can deviate to $b$.
Similarly, if $\pi^a=\{c, s_1, s_2\}$ then $\pi^b=\emptyset$, and $c$ and $s_1$ can deviate to $b$.
If follows that $|\pi^a|\in[4, 3k+3]$ and hence $|W^{\prime}|\le 3k+2$;
also, either $\pi^b=\{s_2\}$ or $\pi^b=\emptyset$.
Now, if $W^{\prime}$ is not a hitting set for $(W, \calG)$,
there exists a set $W_i\in \calG$ such that all players in $W_i$ are assigned to the void activity.
Then if $\pi^b=\emptyset$ the coalition $\{c\}\cup W_i \cup D_i$
together with the activity $b$ strongly blocks $\pi$, and
if $\pi^b=\{s_2\}$ the coalition $\{c, s_2\}\cup W_i \cup D_i\setminus \{d_i^{(1)}\}$ 
together with the activity $b$ strongly blocks $\pi$.
In either case, we obtain a contradiction with the stability of $\pi$.
Thus, $W^\prime$ is a hitting set for $(W, \calG)$. At least one of the three sets
$X^\prime=\{x_v\mid x_v\in W^\prime\}$, $Y^\prime=\{y_v\mid y_v\in W^\prime\}$, 
and $Z^\prime=\{z_v\mid z_v\in W^\prime\}$ contains at most $k$ elements;
assume without loss of generality that $|X^\prime|\le k$. 
By construction, $\{v\in V\mid x_v\in X^\prime\}$ is a hitting set for $(V, \calH)$.
Thus, $(V, \calH)$ admits a hitting set of size at most $k$, which is what we wanted to prove.
\end{proof}

The hardness result in Theorem~\ref{thm:NP:star:core}
immediately generalizes to instances of \gGASP\
with more than two activities: we can modify the construction in our hardness reduction
by introducing additional activities that no player wants to engage in.

Also, Theorem~\ref{thm:NP:star:core} can be extended from stars to cliques;
however, our proof for cliques relies on having at least four non-void activities.
It remains an interesting open problem whether core stable outcomes of \gGASP s on cliques
can be found efficiently if the number of activities does not exceed $3$.
We conjecture that the answer is `no', i.e., the problem of computing core stable outcomes
remains NP-hard for $|A^*|=2, 3$.

\begin{theorem}\label{thm:NP:clique:core}
It is {\em NP}-complete to determine whether an instance of \gGASP\ admits a core stable assignment even when the underlying 
graph is a clique and the number of activities is $4$.
\end{theorem}

In contrast, checking the existence of core stable assignments in
\gGASP\ is easy if $|A^*|=1$, irrespective of the structure of the social network. 

\begin{proposition}\label{prop:core-single}
Every instance of \gGASP\ with $A=\{a, a_\emptyset\}$ 
admits a core stable assignment; moreover, 
such an assignment can be computed in polynomial time.
\end{proposition}
\begin{proof}
Consider an instance $(N, (\succeq_i))_{i\in N}, A, L)$ of \gGASP\ with $A=\{a, a_\emptyset\}$, 
and let $n=|N|$. For each $s\in[n]$, let $S_s$
be the set of all players who weakly prefer $(a, s)$ to $(a_\emptyset, 1)$.
If for each $s\in[n]$ the largest connected component of $S_s$ 
with respect to $(N, L)$ contains 
fewer than $s$ agents, then no outcome in which a non-empty set 
of agents engages in $a$ is individually rational, whereas the assignment
$\pi$ given by $\pi(i)=a_\emptyset$ for all $i\in N$ is core stable.
Otherwise, consider the largest value of $s$ such that $S_s$
has a connected component of size at least $s$. Find a connected subset of $S_s$
of size exactly $s$, and assign the agents in this subset to $a$;
assign the remaining agents to $a_\emptyset$. To see that this assignment is core stable, 
note that for every deviating coalition $S$ we would have $|S|>s$, and hence 
such a coalition is either disconnected or some players in $S$ prefer $(a_\emptyset, 1)$
to $(a, |S|)$.  
\end{proof}

It is not clear if the problem of finding core stable assignments in \gGASP\ 
is in FPT with respect to the number of activities
when the underlying graph is a path. However, we can place this problem in XP
with respect to this parameter. 
In fact, we have the following more general result 
for graphs with few connected coalitions (see the work of Elkind~\cite{Elkind2014}
for insights on the structure of such graphs).

\begin{proposition}\label{prop:core-xp}
Given an instance $(N, (\succeq_i))_{i\in N}, A, L)$
of \gGASP\ with $|N|=n$, $|A|=p+1$, 
such that the number of non-empty connected subsets of $(N, L)$ 
is $\kappa$, we can decide in time $O(\kappa^p)\cdot\poly(n, p)$
whether this instance admits a core stable assignment,
and find one such assignment if it exists.
\end{proposition}
\begin{proof}
Let $C_1, \dots, C_\kappa$ be the list of all non-empty connected subsets
of $(N, L)$. Since each assignment of players to activities
has to assign a connected (possibly empty) subset of players to each activity, 
we can bound the number of possible assignments by $(\kappa+1)^p$:
each of the $p$ non-void activities is assigned to a coalition in our list or to no player at all, 
and the remaining players are assigned the void activity.
We can then generate all these assignments one by one and check
if any of them is core stable; by Proposition~\ref{prop:in-core}, 
the stability check can be performed in time polynomial in $n$ and $p$.
\end{proof}

If the social network $(N, L)$ is a path, it has $O(|N|^2)$ connected subsets.
Thus, we obtain the following corollary.

\begin{cor}\label{cor:core-xp-path}
The problem of deciding whether a given instance
of \gGASP\ whose underlying graph is a path admits
a core stable assignment is in {\em XP} with respect to the number of activities. 
\end{cor}


\newcommand{\MCC}{\textsc{Multicolored Clique}}

\section{Few Players}
\noindent
In the previous sections, the parameter that we focused on
was the number of activities $p$.
Although we expect this parameter to be small in many realistic settings, there
are also situations where players can choose from a huge variety of possible activities.
It is then natural to ask if stability-related problems for \gGASP s
are tractable in the number of players $n$ is small.
This is the question that we consider in this section.

We first observe that for all stability concepts considered in this paper
the problem of finding a stable feasible assignment is in XP with respect to $n$:
we can simply guess the activity of each player (there are at most~$(p+1)^n$ possible guesses) 
and check whether the resulting assignment is feasible and stable.
\begin{observation}\label{obs:XPplayers}
The problem of deciding whether a given instance
of \gGASP\ admits a core stable, Nash stable, or individually stable assignment
is in {\em XP} with respect to the number of players $n$. 
\end{observation}

The following theorem shows that \gGASP{} is unlikely to be fixed-parameter tractable with respect to $n$.
This is somewhat surprising, because an FPT algorithm with respect to~$n$ could afford 
to iterate through all possible partitions of the players into coalitions.
Interestingly, in the following hardness reduction there is a unique 
partition of players into coalitions that can lead to a stable assignment;
that is, computational hardness comes solely from assigning known
coalitions of players to activities.

\begin{theorem}\label{thm:W1players:clique}
The problem of determining whether an instance of \gGASP\ whose underlying graph is a clique 
admits a core stable (Nash stable, or individually stable) assignment is {\em W[1]}-hard
with respect to the number of players $n$. 
\end{theorem}
\begin{proof}
Due to space constraints, we only provide a proof for Nash stability.
We describe a parameterized reduction from the W[1]-hard \MCC{} problem.
An instance of this problem is given by an undirected graph~$G=(V,E)$, a positive integer~$h\in \mathbb{N}$,
and a vertex coloring~$\phi \colon V\to \{1, 2, \ldots, h\}$. It is a `yes'-instance
if~$G$ admits a colorful $h$-clique, that is,
a size-$h$ vertex subset~$Q\subseteq V$ such that the vertices in $Q$ are pairwise adjacent
and have pairwise distinct colors.
Without loss of generality, we assume that there are exactly $q$ vertices of each color for some $q\in\mathbb N$, 
and that there are no edges between vertices of the same color.

Let $(G,h,\phi)$ be an instance of \MCC{} with $G=(V, E)$.
For every $i \in [h]$,
we write $V^{(i)} = \{v_1^{(i)}, \ldots, v_q^{(i)}\}$ to denote the set of vertices of color~$i$.
For each vertex $v$, we write $E(v) = \{e_1^{(v)}, \ldots, e_{|E(v)|}^{(v)}\}$ to denote the set of edges incident to~$v$.
We construct our \gGASP\ instance as follows.
We have one \emph{vertex activity}~$a_\ell^{(i)}$ for each vertex~$v_\ell^{(i)} \in V^{(i)}$, $i \in [h]$,
one \emph{edge activity}~$b_\ell^{(v)}$ for each edge~$e_\ell^{(v)} \in E$,
two \emph{color activities}~$c_i$ and~$c'_i$ for each color~$i \in [h]$, and
one \emph{colorpair activity}~$c_{i,j}$ for each color pair~$i,j \in [h], i \neq j$.

\emph{Idea.}
We will have one \emph{color gadget} for each color~$i \in [h]$ and one \emph{colorpair gadget} for each color pair~$\{i,j\},i \neq j$.
For most of the possible assignments, these gadgets will be unstable, unless the following holds:
\begin{enumerate}
 \item For each color~$i \in [h]$ the first two players from the color gadget select
       a vertex of color~$i$ (by being assigned together to the corresponding vertex activity).
 \item For each color pair~$\{i,j\} \in [h], i \neq j$ the first two players of the colorpair gadget select 
       an edge connecting one vertex of color~$i$ and one vertex of color~$j$ (by being assigned together to
       the corresponding edge activity).
 \item Every selected edge for the color pair~$\{i,j\}$ must be incident to both vertices selected for color~$i$ and color~$j$. 
\end{enumerate}
If the three conditions above hold, then the assignment must encode a colorful $h$-clique.

\emph{Construction details.}
The color gadget~$C(i)$, $i\in [h]$ consists of the following four players.
 \begin{align*}
p_1^{(i)}:&~ (a_1^{(i)},2) \succ B(1,i,3) \succ (a_2^{(i)},2) \succ B(2,i,3) \succ \\ &~ \dots \succ (a_q^{(i)},2) \succ B(q,i,3) \succ (c_i,1) \succ (a_{\emptyset},1)\\
p_2^{(i)}:&~ (a_q^{(i)},2) \succ B(q,i,3) \succ (a_{q-1}^{(i)},2) \succ B(q-1,i,3) \succ \\ &~ \dots \succ (a_1^{(i)},2) \succ B(1,i,3) \succ (c'_i,1) \succ (a_{\emptyset},1)\\
p_3^{(i)}:&~ (c_i,2) \succ (a_{\emptyset},1)\\
p_4^{(i)}:&~ (c'_i,2) \succ (a_{\emptyset},1)\\
  \text{with} &~ B(\ell,i,z) = (b_1^{(v_\ell^{(i)})},z) \succ (b_2^{(v_\ell^{(i)})},z) \dots \succ (b_{|E(v_\ell^{(i)})|}^{(v_\ell^{(i)})},z)
\end{align*}

The colorpair gadget~$C(\{i,j\})$, $\{i,j\} \subseteq [h], i \neq j$, consists of three players.
Players~$p_1^{\{i,j\}}$ and~$p_2^{\{i,j\}}$ approve all edge activities 
corresponding to edges between vertices of color~$i$ and~$j$ with size two
and the colorpair activity~$c_{i,j}$ with size two and one, and strictly prefer the former alternatives to the latter.
Player $p_3^{\{i,j\}}$ approves only the colorpair activity~$c_{i,j}$ with size three.

\emph{Correctness.}
We will now argue that the graph $G$ admits a colorful clique of size $h$ 
if and only if our instance of \gGASP\ admits a Nash stable assignment.

Suppose that there exists a colorful clique~$H$ of size $h$.
Assign the first two players of the color gadget~$C(i)$ to the activity corresponding to the vertex of color~$i$ from~$H$.
Assign the first two players of the colorpair gadget~$C(\{i,j\})$ to the activity corresponding
to the edge between the vertices of color~$i$ and~$j$ in~$H$.
Assign all other players to the void activity.
Clearly, the resulting assignment is feasible.
Observe that for a successful Nash deviation a player must join an existing non-empty coalition, because no player
prefers a size-one activity to the currently assigned one.
By construction, the last two players of each color gadget and the last player of the colorpair gadget cannot deviate
(no other players engage in an approved activity).
Consider the first two players of color gadget~$C(i)$.
They cannot deviate to a vertex activity, because their current activity is their only approved vertex activity 
that has some players assigned to it.
They cannot deviate to an edge activity either, because they would only prefer edge activities corresponding to
edges that are not incident to the vertex of color~$i$ that is part of colorful clique~$H$; these activities, however, 
have no players assigned to them.
Finally, the first two players of a colorpair gadget~$C(\{i,j\})$ 
cannot deviate since no player is assigned to any of their more preferred activities.
Thus, we have a Nash stable feasible assignment.

Suppose that there exists a Nash stable feasible assignment~$\pi$.
Observe that the first two players of each colorpair gadget~$C(\{i,j\})$ are assigned to the same edge activity
corresponding to some edge $e_{\ell_{i,j}} \in E$.
Otherwise, they would need to be assigned to 
the colorpair activity~$c_{i,j}$ and be stalked by the third player of the gadget (similarly
to the Stalker game from Example~\ref{ex:NS:empty}).
We say that these players ``select edge~$e_{\ell_{i,j}}$''.
By similar reasoning the first two players of each color gadget~$C(i)$ are assigned to the same vertex activity
corresponding to some vertex~$v_{\ell_i}^{(i)}$. 
They cannot be assigned to an edge activity, because the third player would 
necessarily be a player from a colorpair gadget, and
these players do not want to engage in activities with two other players.
We say that these players ``select vertex~$v_{\ell_i}^{(i)}$''.
(Note that at this point, the coalition structure of any stable outcome is already fixed:
The first two players of each color gadget and of each colorpair gadget
must form a coalition of size two, respectively, and
all other players must be assigned to the void activity.)

Now, assume towards a contradiction that the selected vertices and edges do not form a colorful clique of size~$h$.
The size and colorfulness are clear from the construction.
Hence, there must be some pair~$\{v_{\ell_i}^{(i)},v_{\ell_j}^{(j)}\}$ of selected vertices that is not adjacent.
However, this would imply that colorpair gadget~$\{i,j\}$ selected an edge 
that is either not incident to vertex~$v_{\ell_i}^{(i)}$
or not incident to vertex~$v_{\ell_j}^{(j)}$.
Without loss of generality let it be non-adjacent to vertex~$v_{\ell_i}^{(i)}$.
Now, there are two cases.
First, if $e_{\ell_{i,j}} \in E(v_{x}^{(i)})$ with $x<\ell_i$, then
player~$p_1^{(i)}$ would deviate to the activity corresponding to~$e_{\ell_{i,j}}$.
Second, if $e_{\ell_{i,j}} \in E(v_{x}^{(i)})$ with $x>\ell_i$,
then  player~$p_2^{(i)}$ would deviate to the activity corresponding to~$e_{\ell_{i,j}}$.
In both cases we have a contradiction to the assumption that~$\pi$ is Nash stable.
\end{proof}

Note that although we showed W[1]-hardness for each of the parameters~$p$ and $n$,
parameterizing by the combined parameter~$p+n$ immediately gives fixed-parameter tractability,
since the input size is trivially upper-bounded by $n^2\cdot p$.

\section{Conclusion}
\noindent
We have investigated the parameterized complexity of computing stable
outcomes of group activity selection problems on networks, with respect to two natural parameters.
Many of our hardness results hold for the classic \GASP\ problem, where there are no constraints
on possible coalitions; however, some of our positive results only hold for acyclic graphs.
Interestingly, one of our tractability results holds for \GASP s, but it is not clear
if it can be extended to \gGASP s; thus, while simple networks may decrease complexity, 
allowing for arbitrary networks may have the opposite effect.

In some of our hardness reductions players have non-strict preferences.
In some cases, it is easy to modify our gadgets to rely on strict preferences only;
in other cases (notably, in individual stability proofs) our arguments make crucial use 
of indifferences. 

Confirming our intuition, we showed that finding core stable assignments 
tends to be computationally more difficult than finding Nash stable or individually stable assignments;
we note that this distinction only becomes visible from the parameterized complexity perspective, 
as all of these problems are NP-hard.

Our results for games with few players are somewhat preliminary:
we have only investigated the complexity of our problems with respect to
this parameter for general networks and for cliques, and have not considered
simple networks such as stars and paths. These questions offer an interesting direction for future work.


\end{document}